\documentclass[11pt]{gen-j-l}

\usepackage{amsmath,amsthm,amssymb,amsfonts,verbatim}
\usepackage[top=1.2in, bottom=1.2in, left=1.2in, right=1.2in]{geometry}

\usepackage[dvipsnames,usenames]{color}
\usepackage[pagebackref]{hyperref}

 \providecommand{\F}{\mathbb{F}}

\title[Subspace Designs]{Subspace Designs based on Algebraic Function Fields}
\author{Venkatesan Guruswami}
\address{Computer Science Department, Carnegie Mellon University, Pittsburgh, USA.}
\email{guruswami@cmu.edu}
\thanks{Research supported in part by NSF CCF-1422045.}
\author{Chaoping Xing}
\address{School of Physical and Mathematical Sciences, Nanyang Technological University, Singapore.}
\email{xingcp@ntu.edu.sg}
\author{Chen Yuan}
\address{School of Physical and Mathematical Sciences, Nanyang Technological University, Singapore.}
\email{yuan0064@e.ntu.edu.sg}

\date{}

\newtheorem{lemma}{Lemma}[section]
\newtheorem{theorem}[lemma]{Theorem}
\newtheorem{cor}[lemma]{Corollary}

\newtheorem{defn}{Definition}

\theoremstyle{remark}
\newtheorem{rmk}{Remark}

\newcommand{\eps}{\varepsilon}
\renewcommand{\epsilon}{\varepsilon}
\renewcommand{\le}{\leqslant}
\renewcommand{\ge}{\geqslant}

\def\ZZ{\mathbb{Z}}
\def\PP{\mathbb{P}}

\def \mH {\mathcal{H}}

\def \mF {\mathcal{F}}

\def \mI {\mathcal{I}}

\def \mL {\mathcal{L}}

\def \mV {\mathcal{V}}
\def \mW {\mathcal{W}}

\def\Pin{{P_{\infty}}}

\def \Xi {{X^{[i]}}}

\newcommand{\Ga}{\alpha}
\newcommand{\Gb}{\beta}
\newcommand{\Gg}{\gamma}     
\newcommand{\Gd}{\delta}     
\newcommand{\Ge}{\epsilon}

\newcommand{\Gl}{\lambda}    \newcommand{\GL}{\Lambda}
     \newcommand{\GO}{\Omega}
\newcommand{\Gs}{\sigma}

\def\g{{\mathfrak{g}}}

\def\bg{{\bf g}}

\def \bz {{\bf z}}

\def\supp {{\rm supp }}

\def\Aut {{\rm Aut }}
\def\End {{\rm End }}
\def\Gal {{\rm Gal }}

\newcommand{\dm}

\begin{document}

\maketitle
\begin{abstract}
Subspace designs are a (large) collection of high-dimensional subspaces $\{H_i\}$ of $\F_q^m$ such that for any low-dimensional subspace $W$, only a small number of subspaces from the collection have non-trivial intersection with $W$; more precisely, the sum of dimensions of $W \cap H_i$ is at most some parameter $L$. The notion was put forth by Guruswami and Xing (STOC'13) with applications to list decoding variants of Reed-Solomon and algebraic-geometric codes, and later also used for explicit rank-metric codes with optimal list decoding radius.

Guruswami and Kopparty (FOCS'13, Combinatorica'16) gave an explicit construction of subspace designs with near-optimal parameters. This construction was based on polynomials and has close connections to folded Reed-Solomon codes, and required large field size (specifically $q \ge m$). Forbes and Guruswami (RANDOM'15) used this construction to give explicit constant degree ``dimension expanders" over large fields, and noted that subspace designs are a powerful tool in linear-algebraic pseudorandomness.

Here, we construct subspace designs over any field, at the expense of a modest worsening of the bound $L$ on total intersection dimension. Our approach is based on a (non-trivial) extension of the polynomial-based construction to algebraic function fields, and instantiating the approach with cyclotomic function fields.  Plugging in our new subspace designs in the construction of Forbes and Guruswami yields dimension expanders over $\F^n$ for any field $\F$, with logarithmic degree and expansion guarantee for subspaces of dimension $\Omega(n/(\log \log n))$.

\end{abstract}

\section{Introduction}

An emerging theory of ``linear-algebraic pseudorandomness'' studies
the linear-algebraic analogs of fundamental Boolean pseudorandom
objects where the rank of subspaces plays the role of the size of
subsets. A recent work~\cite{FG-random15} studied the
interrelationships between several such algebraic objects such as
subspace designs, dimension expanders, rank condensers, and
rank-metric codes, and highlighted the fundamental unifying role
played by \emph{subspace designs} in this web of connections.

Informally, a subspace design is a collection of subspaces of a vector
space $\F_q^m$ (throughout we denote by $\F_q$ the finite field with
$q$ elements) such that any low-dimensional subspace $W$ intersects
only a small number of subspaces from the collection. More precisely:
\begin{defn}
A collection $H_1,H_2,\dots,H_M$ of $b$-dimensional subspaces of $\F_q^m$ form an $(s,L)$-(strong) subspace design, if for every $s$-dimensional subspace $W \subset \F_q^m$, $\sum_{i=1}^M \dim(W \cap H_i) \le L $.
\end{defn}
In particular, this implies that at most $L$ subspaces $H_i$ have non-trivial intersection with $W$. A collection meeting this weaker requirement is called a \emph{weak} subspace design; unless we mention otherwise, by subspace design we always mean a strong subspace design in this paper. One would like the dimension $b$ of each subspace in the subspace design to be large, typically $\Omega(m)$
or applications of interest, $L$ to be small, and the number of subspaces $M$ to be large.

Subspace designs were introduced by the first two authors in \cite{GX-stoc13}, where they used them to improve the list size and efficiency of list decoding algorithms for algebraic-geometric codes, yielding efficiently list-decodable codes with optimal redundancy over fixed alphabets and small output list size. A standard probabilistic argument shows that a random collection of subspaces forms a good subspace design with high probability. Subsequently, Guruswami and Kopparty~\cite{GK-combinatorica} gave an explicit construction of subspace designs, nearly matching the parameters of random constructions, albeit over large fields.

Intriguingly, the construction in \cite{GK-combinatorica} was based on algebraic list-decodable codes (specifically folded Reed-Solomon codes). Recall that improving the list-decodability of such codes was the motivation for the formulation of subspace designs in the first place! This is yet another compelling example of the heavily intertwined nature of error-correcting codes and other pseudorandom objects.  The following states one of the main trade-offs achieved by the construction in \cite{GK-combinatorica}.

\begin{theorem}[Folded Reed-Solomon based construction~\cite{GK-combinatorica}]
\label{thm:intro-GK}
For every $\eps \in (0,1)$, positive integers $s,m$ with $s \le \eps m/4$, and a prime power $q > m$, there exists an explicit\footnote{By explicit, we mean a deterministic construction that runs in time $\text{poly}(q,m,M)$ and outputs a basis for each of the subspaces in the subspace design.}
collection of
$M = q^{\Omega(\epsilon m/ s)}$ subspaces in $\F_q^m$, each of dimension at least $(1-\epsilon) m$,
which form a $(s, \frac{2 s}{\epsilon})$-(strong) subspace design.
\end{theorem}

Note the requirement of the field size $q$ being larger than the ambient dimension $m$ in their construction. To construct subspace designs over small fields, they use a construction over a large extension field $\F_{q^r}$, and view $b$-dimensional subspaces of $\F_{q^r}^{m'}$ as $br$-dimensional subspaces of $\F_q^{rm'}$. However, this transformation need not preserve the ``strongness" of the subspace design, and an $(s,L)$-subspace design over the extension field only yields an $(s,L)$-weak subspace design over $\F_q$.

The strongness property is crucial for all the applications of subspace designs in \cite{FG-random15}. In particular, the strongness is what drives the construction of dimension expanders (defined below) of low degree. The weak subspace design property does \emph{not} suffice for these applications.

\begin{defn}
\label{def:dim-expanders}
A collection of linear maps $A_1,A_2,\dots,A_d : \F^n \to \F^n$ is said to be a $(b,\alpha)$-dimension expander if for every subspace $V$ of $\F^n$ of dimension at most $b$, $\dim(\sum_{i=1}^d A_i(V)) \ge (1+\alpha) \cdot \dim(V)$. The number of maps $d$ is the ``degree" of the expander, and $\alpha$ is the expansion factor.
\end{defn}

Using the subspace designs constructed in Theorem~\ref{thm:intro-GK} in a black-box fashion, Forbes and Guruswami~\cite{FG-random15} gave explicit $(\Omega(n),\Omega(1))$-dimension expanders of $O(1)$ degree when $|\F| \ge \mathrm{poly}(n)$. Here explicit means that the maps $A_i$ are specified explicitly, say by the matrix representing their action with respect to some fixed basis. Extending Theorem~\ref{thm:intro-GK} to smaller fields will yield constant-degree $(\Omega(n),\Omega(1))$-dimension expanders over all fields. The only known constructions of such dimension expanders over finite fields rely on monotone expanders~\cite{DW10,DS11}, a rather complicated (and remarkable) form of bipartite vertex expanders whose neighborhood maps are monotone. Even the existence of constant-degree monotone expanders does not follow from standard probabilistic methods, and the only known explicit construction is a sophisticated one using the group $\text{SL}_2(\mathbb{R})$ by Bourgain and Yehudayoff~\cite{BY13}. (Earlier, Dvir and Shpilka~\cite{DS11} constructed monotone expanders of logarithmic degree using Cayley graphs over the cyclic group, yielding logarithmic degree $(\Omega(n),\Omega(1))$-dimension expanders.)

In light of this, it is a very interesting question to remove the field size restriction in Theorem~\ref{thm:intro-GK} above, as it will yield an arguably simpler construction of constant-degree dimension expanders over every field, and which might also offer a quantitatively better trade-off between the degree and expansion factor. We note that probabilistic constructions achieve similar parameters (in fact a slightly larger sized collection with $q^{\Omega(\eps m)}$ subspaces) with no restriction on the field size (one can even take $q=2$).

\medskip \noindent {\bf Our construction.} The large field size in
Theorem~\ref{thm:intro-GK} was inherited from Reed-Solomon codes,
which are defined over a field of size at least the code length. Our
main contribution in this work is a construction of subspace designs
based on algebraic function fields, which permits us to construct
subspace designs over small fields. By instantiating this approach
with a construction based on cyclotomic function fields, we are able
to prove the following main result in this work:

\begin{theorem}[Main Theorem]\label{thm:1.1}
For every $\epsilon\in (0,1)$, a prime power $q$ and positive integers $s,m$ such that $s\leq \epsilon m/4$, there exists an explicit construction of $M= \Omega(q^{\lfloor \epsilon m/(2s)\rfloor}/\epsilon)$ subspaces in $\F_q^m$, each of dimension at least $(1-\epsilon) m$, which form an $\Bigl(s',\frac{2 s'\lceil\log_q(m)\rceil}{\epsilon}\Bigr)$-strong subspace design for all $s' \le s$.
\end{theorem}

Note that we state a slightly stronger property that the bound on intersection size improves for subspaces of lower dimension $s' \le s$. This property also holds for Theorem~\ref{thm:intro-GK} and in fact is important for the dimension expander construction in \cite{FG-random15}, and so we make it explicit.

The bound on intersection size we guarantee above is worse than the
one from the random construction by a factor of $\log_q m$. The result
of Theorem~\ref{thm:intro-GK} can be viewed as a special case of
Theorem \ref{thm:1.1} since $\log_qm\le 1$ when $q>m$. The factor
$\log_q m$ comes out as a trade-off of the explicit construction vs
the random construction given in \cite{GX-stoc13}.  The extension field
based construction using Theorem~\ref{thm:intro-GK} would yield an $(s,
O(s^2/\eps))$-subspace design (since an $(s,L)$-weak subspace design is
trivially an $(s,sL)$-(strong) subspace design).  The bound we achieve
is better for all $s = \Omega(\log_q m)$. In the use of subspace designs
in the dimension expander construction of \cite{FG-random15}, $s$ governs the dimension of the subspaces which are guaranteed to expand, which we would like to be large (and ideally $\Omega(m)$).
The application of subspace designs to list
decoding~\cite{GX-stoc13,GWX-ieeetit} employs the parameter choice
$m=O(s)$ in order keep the alphabet size $q^m$ small. Therefore, our
improvement applies to a meaningful setting of
parameters that is important for the known applications of (strong) subspace designs.

\medskip \noindent {\bf Application to dimension expanders over small fields.}
By plugging in the subspace designs of Theorem~\ref{thm:1.1} into the dimension expander construction of \cite{FG-random15}, we can get the following:
\begin{theorem}
\label{thm:intro-subs-design}
For every prime power $q$ and positive integer $n \ge q$, there exists an explicit construction of a $\Bigl(b=\Omega\bigl( \frac{n}{\log_q \log_q n} \bigr), 1/3\Bigr)$-dimension expander with $O(\log_q n)$ degree.
\end{theorem}

For completeness, let us very quickly recap how such dimension
expanders may be obtained from the subspace designs of
Theorem~\ref{thm:1.1}, using the ``tensor-then-condense" approach in
\cite{FG-random15}. We begin with linear maps $T_1,T_2 : \F^n \to
\F^{2n}$, where $T_1(v) = (v;0)$ and $T_2(v) = (0;v)$ --- these
trivially achieve expansion factor $2$ by doubling the ambient
dimension. Then we take the subspace design of Theorem~\ref{thm:1.1}
with $m=2n$, $\eps=1/2$, $s=2b$, and $M = 12 \lceil \log_q m \rceil$
subspaces $H_i$ (if $b = \beta n/(\log_q\log_q n)$ for small enough
absolute constant $\beta > 0$, Theorem~\ref{thm:1.1} guarantees these
many subspaces). Let $E_i : \F^{2n} \to \F^n$ be linear maps such that
$H_i = \mathrm{ker}(E_i)$. The dimension expander consists of the $2M$
composed maps $E_i \circ T_j$ for $i=1,2,\dots,M$ and $j=1,2$.
Briefly, the analysis of the expansion in dimension proceeds as
follows. Let $V$ be a subspace of $\F^n$ with $\dim(V) = \ell \le b$,
and let $W = T_1(V) + T_2(V)$ be the $2\ell$-dimensional subspace of
$\F^{2n}$ after the tensoring step. The strong subspace design
property implies that the number of maps $E_i$ for which $\dim(E_i W)
< 4\ell/3$ --- which is equivalent to $\dim(W \cap H_i) > 2\ell/3$ ---
is less than $12 \lceil \log_q m\rceil = M$. So there must be an $i$
for which $\dim(E_i W) \ge 4\ell/3$, and this $E_i$ when composed with
$T_1$ and $T_2$ will expand $V$ to a subspace of dimension at least
$\frac{4}{3} \dim(V)$.

By using a method akin to the conversion of Reed-Solomon codes over
extension fields to BCH codes over the base field, applied to the
large field subspace designs of Theorem~\ref{thm:intro-GK}, Forbes and
Guruswami~\cite{FG-random15} constructed $(\Omega(n/\log n),
\Omega(1))$-dimension expanders of $O(\log n)$ degree. In contrast,
our construction here guarantees expansion for dimension up to
$\Omega(n/(\log \log n))$.  The parameters offered by
Theorem~\ref{thm:intro-subs-design} are, however, weaker than both the
construction given in \cite{DS11}, which has logarithmic degree but
expands subspaces of dimension $\Omega(n)$, as well as the one in
\cite{BY13}, which further gets constant degree. However, we do not go
through monotone expanders which are harder to construct than vertex
expanders, and our construction works fully within the
linear-algebraic setting. We hope that the ideas in this work pave the
way for a subspace design similar to Theorem~\ref{thm:intro-GK} over
small fields, and the consequent construction of constant-degree
$(\Omega(n),\Omega(1))$-dimension expanders over all fields. In fact,
all that is required for this is an $(s,O(s))$-subspace design with a
sufficiently large constant number of subspaces, each of dimension
$\Omega(m)$.


\medskip \noindent {\bf Construction approach.}
The generalization of the polynomials-based subspace design from \cite{GK-combinatorica} to take advantage of more general algebraic function fields is not straightforward. The natural approach would be to replace the space of low-degree polynomials by a Riemann-Roch space consisting of functions of bounded pole order $\ell$ at some place. We prove that such a construction can work, provided the
degree $\ell$ is less than the degree of the field extension (and some other mild condition is met, see Lemma~\ref{lem:3.2}).
However, this degree restriction is a severe one, and the dimension of the associated Riemann-Roch space will typically be too small (as the ``genus" of the function field, which measures the degree minus dimension ``defect," will be large), unless the field size is large. Therefore, we don't know an instantiation of this approach that yields a family of good subspace designs over a fixed size field.

Let us now sketch the algebraic crux of the polynomial based construction in \cite{GK-combinatorica}, and the associated challenges in extending it to other function fields.  The core property of a dimension $s$ subspace $W$ of polynomials underlying the construction of Theorem~\ref{thm:intro-GK} is the following: If $f_1,f_2,\dots,f_s \in \F_q[X]$ of degree less than $q-1$ are linearly independent over $\F_q$ (these $s$ polynomials being a basis of the subspace $W$), then the ``folded Wronskian," which is the determinant of the matrix $M(f_1,f_2,\dots,f_s)$ whose $i,j$'th entry is $f_j(\gamma^{i-1} X)$, is a \emph{nonzero} polynomial in $\F_q[X]$.  Here $\gamma$ is an arbitrary primitive element of $\F_q$. One might compare this with the classical Wronskian criterion for linear dependence over characteristic zero fields (and also holds when characteristic is bigger than the degree of the $f_i$'s), based on the singularity of the $s \times s$ matrix whose $i,j$'th entry is $\frac{d^{i-1} f_j}{d X^{i-1}}$.

One approach is to prove this claim about the folded Wronskian is via a ``list size" bound from list decoding: one can prove that for any $A_1,\dots,A_s \in \F_q[X]$, not all $0$, the space of solutions $f \in \F_q[X]_{< (q-1)}$ to
\begin{equation}
\label{eq:frs-eqn}
A_1(X) f(X) + A_2(X) f(\gamma X) + \cdots + A_s(X) f(\gamma^{s-1} X) = 0
\end{equation}
has dimension at most $s-1$. (This was the basis of the linear-algebraic list decoding algorithm for folded Reed-Solomon codes~\cite{Gur-ccc11,GW-merged}.) Stating the contrapositive, if $f_1,f_2,\dots,f_s$ are linearly dependent over $\F_q[X]$, then the rows of the matrix $M(f_1,f_2,\dots,f_s)$ are linearly independent, and therefore its determinant, the folded Wronskian, is a nonzero polynomial. On the other hand, being the determinant of an $s \times s$ matrix whose entries are degree $m$ polynomials, the folded Wronskian has degree at most $ms$. To prove the subspace design property, one then establishes that for each subspace $H_i$ in the collection that intersects $W=\text{span}(f_1,\dots,f_s)$, the determinant picks up a number of distinct roots each with $\dim(W \cap H_i)$ multiplicity, the set of roots for different intersecting $H_i$ being disjoint from each other. The total intersection bound then follows because the folded Wronskian has at most $ms$ roots, counting multiplicities.

One can try to mimic the above approach for folded algebraic-geometric (AG) codes, with $f^\sigma$ for some suitable automorphism $\sigma$ playing the role of the shifted polynomial $f(\gamma X)$. This, however, runs into significant trouble, as the bound on number of solutions $f$ to the functional equation analogous  to \eqref{eq:frs-eqn}, $A_1 f  + A_2 f^\sigma + \cdots + A_s f^{\sigma^{s-1}} = 0$,
is much higher. The list of solutions is either exponentially large and needs pruning via pre-coding the folded AG codes with subspace-evasive sets~\cite{GX-stoc12}, or it is much bigger than $q^{s-1}$ in the constructions based on cyclotomic function fields and narrow ray class fields where the folded AG codes work directly~\cite{Gur-cyclo-ANT,GX-jcta}.

Let $F/K$ be a function field where the extension is Galois with Galois group generated by an automorphism $\Gs$. We choose the $m$-dimensional ambient space $\mathcal{V} \cong \F_q^m$ to be a carefully chosen subspace of a Riemann-Roch space in $F$ of degree $\ell \gg m$ (specifically, we require $\ell \ge m+2\g$ where $\g$ is the genus). We then establish that if $f_1,f_2,\dots,f_s \in \mathcal{V}$ are linearly independent over $\F_q$, a certain ``automorphism Moore matrix" $M_{\Gs}(f_1,f_2,\dots,f_s)$ (Definition~\ref{def:auto-moore}) is non-singular. The determinant of this Moore matrix is thus a non-zero function in $F$, and this generalizes the folded Wronskian criterion for polynomials mentioned above.

This non-singularity result is proved in two steps. First, we show that for functions in $\mathcal{V}$, linear independence over $\F_q$ implies linear independence over $K$. Then we show that for any $f_1,\dots,f_s \in F$ that are linearly independent over $K=F^\sigma$, the automorphism Moore matrix associated with $\sigma$  is non-singular. With our hands on the non-zero function $\Delta=\mathsf{det}(M_{\sigma}(f_1,f_2,\dots,f_s))$, we can proceed as in the folded Reed-Solomon case --- the part about $\Delta$ picking up many zeroes whenever a subspace in the collection intersects $\text{span}(f_1,\dots,f_s)$ also generalizes. The pole order of $\Delta$, however, is now $\ell s$ instead of $m s$ in the polynomial-based construction. This is the cause for the worse bound
on total intersection dimension in our Theorem~\ref{thm:1.1}.

\medskip \noindent \textbf{Organization.} We begin with a quick review
of background on algebraic function fields in general and cyclotomic
function fields in particular in Section~\ref{sec:2}. We also elaborate on the
the complexity aspects of computing bases of Riemann-Roch spaces and
evaluating functions at high degree places in cyclotomic function
fields --- this implies that our subspace designs can be constructed
in polynomial time. We present and analyze our constructions of subspace designs
from function fields in Section~\ref{sec:3} --- we give two criteria
that enable our construction, Lemmas \ref{lem:3.1}, \ref{lem:3.2},
though the former is the more useful one for us. In
Section~\ref{sec:4}, we instantiate our construction with specific
cyclotomic function fields and derive our main consequence for
subspace designs and establish Theorem~\ref{thm:1.1}. For reasons of space, several of the technical proofs appear are deferred to appendices.

\section{Preliminaries on function fields}\label{sec:2}

{\bf Background on function fields.}
Throughout this paper, $\F_q$ denotes the finite field of $q$ elements. A function field $F$ over $\F_q$ is a field extension over $\F_q$ in which there exists an  element $z$ of $F$ that is transcendental over $\F_q$ such that $F/\F_q(z)$ is a finite extension. $\F_q$ is called the full constant field of $F$ if the algebraic closure of $\F_q$ in $F$ is $\F_q$ itself. In this paper, we always assume that $\F_q$ is the full constant field of $F$, denoted by $F/\F_q$.

Each discrete valuation $\nu$ from $F$ to $\ZZ\cup\{\infty\}$ defines a local ring $O=\{f\in F:\; \nu(f)\ge 0\}$. The maximal ideal $P$ of $O$ is called a {\it place}. We denote  the valuation $\nu$ and the local ring $O$ corresponding to $P$ by $\nu_P$ and $O_P$, respectively. The residue class field $O_P/P$, denoted by $F_P$, is a finite extension of $\F_q$. The extension degree $[F_P:\F_q]$ is called {\it degree} of $P$, denoted by $\deg(P)$.

Let $\PP_F$ denote the set of places of $F$. A divisor $D$ of $F$ is a formal sum $\sum_{P\in\PP_F}m_PP$, where $m_P\in\ZZ$ are equal to $0$ except for finitely many $P$. The degree of $D$ is defined to be $\deg(D)=\sum_{P\in\PP_F}m_P\deg(P)$. We say that $D$ is positive, denoted by $D\ge 0$, if $m_P\ge 0$ for all $P\in\PP_F$. For a nonzero function $f$, the principal devisor $(f)$ is defined to be $\sum_{P\in\PP_F}\nu_P(f)P$. Then the degree of the principal divisor $(f)$ is $0$.  The Riemann-Roch space associated with a divisor $D$, denoted by $\mL(D)$, is defined by
\begin{equation}\label{eq:2.1}
\mL(D):=\{f\in F\setminus\{0\}:\; (f)+D\ge 0\}\cup\{0\}.
\end{equation}
Then $\mL(D)$ is a finite dimensional space over $\F_q$. By the Riemann-Roch theorem \cite{stich-book}, the dimension of $\mL(D)$, denoted by $\dim_{\F_q}(D)$, is lower bounded by $\deg(D)-\g+1$, i.e., $\dim_{\F_q}(D)\ge \deg(D)-\g+1$, where $\g$ is the genus of $F$. Furthermore, $\dim_{\F_q}(D)= \deg(D)-\g+1$ if $\deg(D)\ge 2\g-1$. In addition, we have the following results \cite[Lemma 1.4.8 and Corollary 1.4.12(b)]{stich-book}:
\begin{itemize}
\item[(i)]
If $\deg(D)<0$, then $\dim_{\F_q}(D)= 0$;
\item[(ii)] For a positive divisor $G$, we have $\dim_{\F_q}(D)-\dim_{\F_q}(D-G)\le \deg(G)$, i.e., $\dim_{\F_q}(D-G)\ge \dim_{\F_q}(D)-\deg(G)$.
\end{itemize}

Let $\Aut(F/\F_q)$ denote the set of automorphisms of $F$ that fix every element of $\F_q$, i.e., \[\Aut(F/\F_q)=\{\tau:\; \mbox{$\tau$ is an automorphism of $F$ and $\Ga^{\tau}=\Ga$ for all $\Ga\in\F_q$}\}.\] For a place $P\in\PP_F$ and an automorphism $\Gs\in\Aut(F/\F_q)$, we denote by $P^{\Gs}$ the set $\{f^{\Gs}:\; f\in P\}$. Then $P^{\Gs}$ is a place and moreover we have $\deg(P^{\Gs})=\deg(P)$. The place  $P^{\Gs}$ is called a conjugate place of $P$. $\Gs$ also induces an automorphsim of $\Aut(\F_P/\F_q)$. This implies that there exists an integer $e\ge 0$ such that $\Ga^{\Gs}=\Ga^{q^e}$ for all $\Ga\in\F_P$. $\Gs$ is called the {\it Frobenius} of $P$ if $e=1$, i.e., $\Ga^{\Gs}=\Ga^{q}$ for all $\Ga\in\F_P$. For a place $P$ and  a function $f\in O_P$, we denote by $f(P)$ the residue class of $f$ in $F_P$. Thus, we have $(f(P))^{q^e}=(f(P))^{\Gs}=f^{\Gs}(P^{\Gs})$.

\medskip \noindent {\bf Background on cyclotomic function fields.}
Let $x$ be a transcendental element over $\F_q$ and denote by $K$ the
rational function field $\F_q(x)$. Let $K^{ac}$ be an algebraic
closure of $K$. Denote by $\F_q[x]$ the polynomial ring $\F_q[x]$. Let
$\End(K^{ac})$ be the ring homomorphism from $K^{ac}$ to $K^{ac}$. We
define $\rho_{x}(z)=z^q+xz$ for all $z\in K^{ac}$. For $i\ge 2$, we
define $\rho_{x^i}(z)=\rho_x(\rho_{x^{i-1}}(z))$. For a polynomial
${p(x)}=\sum_{i=0}^na_ix^i\in\F_q[x]$, we define
$\rho_{p(x)}(z)=\sum_{i=0}^na_i\rho_{x^i}(z)$. For simplicity, we
denote $\rho_{p(x)}(z)$ by $z^{p(x)}$. It is easy to see that
$z^{p(x)}\in \F_q[x][z]$ is a $q$-linearized polynomial in $z$ of
degree $q^d$, where $d=\deg({p(x)})$.

For a polynomial ${p(x)}\in\F_q[x]$ of degree $d$, define the set
\begin{equation}\label{eq:2.2}
\GL_{p(x)}:=\{\Ga\in K^{ac}:\; \Ga^{p(x)}=0\}.
\end{equation}
Then $\GL_{p(x)}\simeq \F_q[x]/({p(x)})$ is an $\F_q[x]$-module and it has exactly $q^d$ elements. Furthermore, $\GL_{p(x)}$ is a cyclic $\F_q[x]$-module. For any generator $\Gl$ of $\GL_{p(x)}$, one has $\GL_{p(x)}=\{\Gl^A:\; A\in \F_q[x]/({p(x)})\}$ and $\Gl^A$ is a generator of $\GL_{p(x)}$ if and only if $\gcd(A,{p(x)})=1$. The extension $K(\Gl)=K(\GL_{p(x)})$ is a Galois extension over $K$   with $\Gal(K(\GL_{p(x)})/K)\simeq (\F_q[x]/{p(x)})^*$, where $(\F_q[x]/{p(x)})^*$ is the unit group of the ring $\F_q[x]/({p(x)})$. We use $\Gs_A$ to denote the automorphism of $\Aut(K(\Gl)/K)$ corresponding to $A$, i.e., $\Gl^{\Gs_A}=\Gl^A$.
The size of $(\F_q[x]/{p(x)})^*$ is denoted by $\Phi({p(x)})$. If ${p(x)}$ is an irreducible polynomial of degree $d$ over $\F_q$, we have $\Phi({p(x)})=q^d-1$. In this case, the extension $K(\GL_{p(x)})/K$ is cyclic and $\Gal(K(\GL_{p(x)})/K)\simeq (\F_q[x]/{p(x)})^*\simeq\F_{q^d}^*$.

From now on in this subsection, we assume that ${p(x)}$ is a monic irreducible polynomial of degree $d$ over $\F_q$. The infinite place $\infty$ of $K$ splits into $(q^d-1)/(q-1)$ places of degree $1$ in $K(\GL_{p(x)})$, each having ramification index $q-1$. The zero place of ${p(x)}$ is totally ramified in $K(\GL_{p(x)})/K$. Furthermore, a monic irreducible polynomial $h(x)\neq {p(x)}$ of $\F_q[x]$ is unramified and splits into $s$ places of degree $r\deg(h)$, where $r$ is the order of $h(x)$ in the unit group $(\F_q[x]/{p(x)})^*$ and $s=(q^d-1)/r$. This implies that the zero place of $x$ in totally inert in $K(\GL_{p(x)})/K$ if ${p(x)}\neq x$ is a monic \emph{primitive} polynomial.

\begin{lemma}\label{lem:2.1}\cite{hayes,rosen}
Let ${p(x)}$ be a monic irreducible polynomial of degree $d$ and let
$\Gl$ be a generator of $\GL_{p(x)}$. Then $\Gl$ is a local parameter
of the unique place ${P'}$ of $K(\GL_{p(x)})$ lying over ${p(x)}$,
i.e., $\nu_{{P'}}(\Gl)=1$. Furthermore, let $O_{K(\GL_{p(x)})}$ be the
integral closure of $\F_q[x]$ in $K(\GL_{p(x)})$. Then
$\{1,\Gl,\dots,\Gl^{m-1}\}$ is an integral basis of
$O_{K(\GL_{p(x)})}$ over $\F_q[x]$, where $m=q^d-1$.
\end{lemma}

Let $\infty$ denote the pole place of $x$ in $K$. The following lemma determines the principal divisor of a generator of $\GL_{p(x)}$.
\begin{lemma}\label{lem:2.2} Let  ${p(x)}$ be a monic irreducible polynomial of degree $d$ and let $\Gl$ be a generator of $\GL_{p(x)}$. Then the principal divisor $(\Gl)$ is equal to
\begin{equation}\label{eq:2.3}
(\Gl)=P'+\sum_{i=1}^{d}\sum_{j=1}^{q^{i-1}}((q-1)(d-i)-1)\infty_{ij},
\end{equation}
where $P'$ is the unique place of  $K(\GL_{p(x)})$ lying over the zero of ${p(x)}$ and $\{\infty_{ij}\}_{1\le i\le d,1\le j\le q^{i-1}}$ is the set of all places of $K(\GL_{p(x)})$ lying over $\infty$ of $\F_q(x)$.
\end{lemma}
\begin{proof} Let us first look at the poles of $\Gl$. Write $\Gl^{p(x)}/\Gl=\sum_{i=0}^d{p(x)\brack i}\Gl^{q^i-1}$, where ${p(x)\brack i}$ denotes the coefficient of $\Gl^{q^i-1}$. Then ${p(x)\brack i}$ is a polynomial in $x$ of degree $q^i(d-i)$.
If a place $Q$ of $K(\GL_{p(x)})$ does not lie over $\infty$ of $K$, we claim that $\nu_Q(\Gl)\ge 0$. Otherwise, one would have $\nu_Q(\Gl^{q^d-1})<\nu_Q\left({p(x)\brack i}\Gl^{q^i-1}\right)$ for all $i\le 0\le d-1$.  This is impossible as $\sum_{i=0}^d{p(x)\brack i}\Gl^{q^i-1}=0$.

By \cite[Theorem 3.2]{hayes}, we know that there exists a place $Q$ of
$K(\GL_{p(x)})$ lying over $\infty$ of $K$ such that $\nu_Q(\Gl)=-1$
and $\nu_Q(\Gl^A)=(q-1)(d-i)-1$ for any polynomial $A\in\F_q[x]$ of
degree $i-1\le d-1$. This implies that for a polynomial $A$ of degree
$i-1\le d-1$ with $\gcd(A,p(x))=1$, one has $\nu_R(\Gl)=(q-1)(d-1)-1$
for $R=Q^{\Gs_B}$, where $B$ is the unique polynomial of degree $<d$
satisfying $AB\equiv 1\bmod{p(x)}$. When $A$ runs through all
polynomials in $(\F_q[x]/(p(x)))^*$, $\Gs_B$ runs through all
conjugate places lying over $\infty$.  This means that there are
exactly $q^{i-1}$ places $R$ lying over $\infty$ with
$\nu_R(\Gl)=(q-1)(d-i)-1$ since there are $q^{i-1}$ monic polynomials
of degree $i-1$ in $(\F_q[x]/(p(x)))^*$. Hence, the divisor
$\sum_{i=1}^{d}\sum_{j=1}^{q^{i-1}}((q-1)(d-i)-1)\infty_{ij}$ appears
as part of the principal divisor $(\Gl)$. The desired result follows
from the following facts: (i) $\Gl$ has no poles other than those
lying $\infty$; (ii) $\Gl$ is a local parameter of $P'$; and (iii)
$\deg\left(\sum_{i=1}^{d}\sum_{j=1}^{q^{i-1}}((q-1)(d-i)-1)\infty_{ij}\right)=-d$. This
completes the proof.
\end{proof}

Now we show that every element in the Riemann-Roch space $\mL(\ell P')$ has a unique representation of certain form.

\begin{lemma}\label{lem:2.3} Let  ${p(x)}$ be a monic irreducible polynomial of degree $d$ and let $\Gl$ be a generator of $\GL_{p(x)}$. Let $P'$ be the place of  $K(\GL_{p(x)})$ lying over ${p(x)}$. Then every nonzero element $f$ of $\mL(\ell P')$ can be uniquely written as $f={p(x)}^{-e}\sum_{i=0}^{m-1}A_i\Gl^i$
for some $e\ge 0$, where $A_i$ are polynomials of $\F_q[x]$ and not all of them are divisible by ${p(x)}$. Furthermore, $\deg(A_i)\le (m-1)/(q-1)+de+d/2$ for all $0\le i\le m-1$.
\end{lemma}

\begin{proof} If $f\in\F_q$, it is clearly true. Now let $f\in\mL(\ell P')\setminus\F_q$. Let $\nu_{P'}(f)=-r<0$ and put $e=\lceil \frac{r}m\rceil$. Then $0\le \nu_{P'}({p(x)}^ef)<m$. Thus, ${p(x)}^ef$ belongs to $O_{K(\GL_{p(x)})}$. By Lemma \ref{lem:2.1}, there exists a set $\{A_i\}_{i=0}^{m-1}$ of $\F_q[x]$ such that ${p(x)}^ef=\sum_{i=0}^{m-1}A_i\Gl^i$. We claim that not all $A_i$ are divisible by ${p(x)}$. Otherwise we would have $\nu_{P'}({p(x)}^{-e}\sum_{i=0}^{m-1}A_i\Gl^i)\ge -em+m>-r$ and this is a contradiction.

Put $g={p(x)}^ef$. Let $\Gs$ be a generator of the Galois group $\Gal(K(\GL_{p(x)})/K)$. Define $\bg=(g,g^{\Gs},\dots,g^{\Gs^{m-1}})$. Since $g^{\Gs^k}={p(x)}^ef^{\Gs^k}\in {p(x)}^e\mL(\ell P')$, we have $\nu_{\infty_{ij}}(g^{\Gs^k})\ge \nu_{\infty_{ij}}({p(x)}^e)\ge -(q-1)de$ for all $0\le k\le m-1$ and each infinite place $\infty_{ij}$. Let $C$ be the $m\times m$ matrix with $(k,l)$ entry equal to $\Gs^k(\Gl^l)$. Let $C_i$ be the matrix obtained from $C$ by replacing  the $i's$ column with the column vector $\bg$. Then we have $A_i=\det(C_i)/\det(C)$.

Since $\nu_{\infty_{ij}}(\Gs^k(\Gl^l))\ge -1$, we have $\nu_{\infty_{ij}}(\det(C_i))\ge -(m-1)-(q-1)de$. As $\det(C)^2=\pm P$, we have $\nu_{\infty_{ij}}(\det(C))=-(q-1)d/2$. Thus, we have  $\nu_{\infty_{ij}}(A_i)=\nu_{\infty_{ij}}(\det(C_i)/\det(C))=-(m-1)-(q-1)de+(q-1)d/2$. This implies that $\deg(A_i)\le (m-1)/(q-1)+de+d/2$.
The proof is completed.
\end{proof}

We next discuss how to evaluate a function at a
place of higher degree. Let $g(x)$ be an irreducible polynomial of
degree $r$ and it splits completely in $K(\GL_{p(x)})$. By the Kummer
Theorem~\cite[Theorem III.3.7]{stich-book}, the polynomial
$\Gl^{p(x)}/x$ is factorized into $m$ product of linear factors over
$\F_q[x]/(p(x))\simeq\F_{q^r}$. Let $\Gl-\Ga$ be a linear factor, then
there is a place $Q$ of degree $r$ of $K(\GL_{p(x)})$. To evaluate a
function $f(x,\Gl)\in \mL(\ell P')$ at $Q$, we can simple compute
$f(\bar{x},\Ga)$, where $\bar{x}$ is the residue class of $x$ in
$\F_q[x]/(p(x))$. It is clear that the complexity of this evaluation
takes time $\text{poly}(q,m,r)$. The above analysis gives the
following result.

\begin{lemma}\label{lem:2.4} Let  ${p(x)}$ be a monic irreducible polynomial of degree $d$ and let $\Gl$ be a generator of $\GL_{p(x)}$. Let $P'$ be the place of  $K(\GL_{p(x)})$ lying over ${p(x)}$. Then  for a place $Q$ of  $K(\GL_{p(x)})$ of degree $r$ that is completely splitting over $K$, the evaluation of a function of
$f(x,\Gl)\in \mL(\ell P')$ at $Q$ can be computed in $\text{poly}(q,m,r)$ time.
\end{lemma}

\noindent {\bf Computing bases.} Our next goal is the following claim, which states that bases for the requisite bases for our construction can be efficiently computed.

Assume that $p(x)$ is a monic primitive polynomial of degree $d$ in $\F_q[x]$. Let $\Gl$ be a generator of  $\GL_{p(x)}$. Then we have the following facts:
\begin{itemize}
\item Every nonzero function $f$ in $\mL(D)$ has the form
\begin{equation}\label{eq:5.1}
f=p(x)^e\sum_{i=0}^{m-1}a_i(x)\Gl^i,
\end{equation}
where $e\ge 0$ and $a_i(x)\in\F_q[x]$ and not all $a_i(x)$ are divisible by $p(x)$.
\item The principal divisor $(\Gl)$ is
\begin{equation}\label{eq:5.2}
(\Gl)=P'+\sum_{i=1}^{d}\sum_{j=1}^{q^{i-1}}((q-1)(d-i)-1)\infty_{ij},
\end{equation}
where $\{\infty_{ij}\}_{1\le i\le d,1\le j\le q^{i-1}}$ is the set of all places lying over $\infty$ of $\F_q(x)$.
\end{itemize}

Let $f$ be a function given in \eqref{eq:5.1}. To show that $f$ belongs $\mL(D)$, it is sufficient to check that $\nu_{P'}(f)\ge -\ell/d$ and $\nu_{\infty_{ij}}(f)\ge 0$ for all places $i,j$.

 Let $i_0$ be the smallest number in $[0,m-1]$ such that $a_i(x)$ is not divisible by $p(x)$.  Then we have  $\nu_{P'}(f)=-em+i_0$. Thus, we have $-em+i_0\ge -\left\lceil\frac{2\g+m-1}d\right\rceil=-m+\left\lfloor\frac{m-q+1}{d(q-1)}\right\rfloor$. This implies that either $e=0$ (in this case $f\in\F_q$) or $e=1$ and $i_0\ge \left\lfloor\frac{m-q+1}{d(q-1)}\right\rfloor$.

To consider $\nu_{\infty_{ij}}(f)$, we note that $t_i:=(x^{d-i}\Gl)^{-1}$ is a local parameter of $\infty_{ij}$ for all $i,j$. Assume that $0=\Gl^{p(x)}/\Gl=\Gl^m+c_{m-1}(x)\Gl^{m-1}+\cdots+c_1(x)\Gl+c_0(x)\in\F_q[x][\Gl]$. Then we get an equation
\begin{equation}\label{eq:5.2a}
x^{-m(d-i)}t_i^{-m}+x^{-(m-1)(d-i)}c_{m-1}(x)t_i^{-(m-1)}+\cdots+x^{-(d-i)}t_i^{-1}+c_0(x)=0.
\end{equation}
Let the local expansion of $x$ at ${\infty_{ij}}$ be \begin{equation}\label{eq:5.3}
\sum_{k=1-q}^{-1}\Ga_kt_i^{-k}+a(t_i)\end{equation}
 for some $\Ga_i\in\F_q$ and $a(x)\in\F_q[x]$.
 Substituting $x$ with local expansion of \eqref{eq:5.3} into \eqref{eq:5.2a} to solve $\Ga_k$. Then  substituting  \eqref{eq:5.3} into \eqref{eq:5.1} get
\begin{equation}\label{eq:5.4} f=\sum_{k=-r}^{-1}\Gb_kt_i^{-k}+b(t_i)\end{equation}
for some integer $r\ge 1$, $\Gb_k\in \F_q$ and $b(x)\in\F_q[x]$. Note that $\Gb_k$ is a linear combination of coefficients of $a_i(x)$.

The genus of the function field $K(\GL_{p(x)})$ is $\g=\frac12\left(d-2+\frac{q-2}{q-1}\right)(q^d-1)+1$. Put
$D=\left\lceil\frac{2\g+m-1}d\right\rceil P'.$ Let $Q'$ be the unique place of $K(\GL_{p(x)})$ lying over $x$.
It is clear that $\ell=\deg(D)\ge 2\g+m$. Thus, $\dim_{\F_q}(D-Q')=\deg(D)-m-\g+1$. Choose $\mV\subseteq\mL(D)$ such that $\mV$ and $\mL(D-Q')$ are a direct sum of $\mL(D)$.

In conclusion, $f$ in the form \eqref{eq:5.1} belongs to $\mL(D)$ if and only if (i) (a) $f\in\F_q$ or (b) $e=1$ and $a_i(x)$ is divisible by $p(x)$ for all $0\le i<\left\lfloor\frac{m-q+1}{d(q-1)}\right\rfloor$; (ii) the local expansion of $f$ in \eqref{eq:5.4} satisfies $\Gb_k=0$ for all $-r\le i\le d(q-1)+1$. Furthermore, $f$ in the form \eqref{eq:5.1} belongs to $\mL(D-Q')$ if and only if in addition $f$ satisfies that $a_i(x)$ is divisible by $x$ for all $0\le i\le m-1$.

To determine $f$, it is equivalent to finding $a_i(x)$. We can solve $a_i(x)$ through a homogenous equation system of about $m^2$ variables that are coefficients of $a_i(x)$. Therefore, one can find a basis of $\mV$ in $\text{poly}(q,m)$ time. Summering the above analysis gives us Lemma~\ref{lem:2.5}

\begin{lemma}\label{lem:2.5} Let  ${p(x)}$ be a monic primitive polynomial of degree $d$ and let $\Gl$ be a generator of $\GL_{p(x)}$. Let $P',Q'$ be the places of  $K(\GL_{p(x)})$ lying over ${p(x)}$ and $x$, respectively. Put $D=\left\lceil\frac{2\g+m-1}d\right\rceil P'$ with $m=q^d-1$. Then a basis of a vector space $\mV$ satisfying $\mV\oplus\mL(D-Q')=\mL(D)$ can be computed in $\text{poly}(q,m)$ time.
\end{lemma}

\section{Construction of subspace design}\label{sec:3}
\subsection{Moore determinant}
The main purpose of this subsection is to provide a function, namely the determinant of a ``Moore" matrix, that is guaranteed to be non-zero when $s$ functions $f_1,f_2,\dots,f_s$ in a function field $F/K$ are linearly independent over $\F_q$. This will provide the necessary generalization of the fact that the folded Wronskian is non-zero when $f_1,\dots,f_s \in \F_q[X]$ of degree less than $(q-1)$ are linearly independent over $\F_q$.

\begin{lemma}\label{lem:3.1}
Let  $F/K$ be a finite field extension.    Suppose that  $Q'$ is a place of $F$ lying above a rational place $Q$ of $K$. Let $D$ be a positive devisor of $F$ with $Q'\not\in\supp(D)$. If $\mV$ is an $\F_q$-subspace of $\mL(D)$ such that $\mV\cap\mL(D-Q')=\{0\}$,
then $f_1,\ldots,f_s\in \mV$ are $\F_q$-linearly independent if and only if they are linearly independent over $K$.
\end{lemma}
\begin{proof} The ``if" part is clearly true. Now assume that  $f_1,\ldots,f_s\in \mV$ are $\F_q$-linearly independent. Suppose that there exist functions $A_1,\ldots,A_s\in K$ such that not all of them are equal to $0$ and
\begin{equation}\label{eq:3.1}
\sum_{i=1}^{s}A_if_i=0.
\end{equation}
By the Strong Approximation Theorem \cite[Theorem I.6.4]{stich-book},
we can multiply $A_i$ with a common nonzero function $B$ in $K$ such
that the only possible pole of $A_iB$ is $Q$ for all
$i=1,2,\dots,s$. Thus, without loss of generality, we may assume that
$\nu_P(A_i)\ge 0$ for all places $P\neq Q$ of $K$.  Let
$a=\max\{-\nu_Q(A_i):\; A_i\neq0,\; 1\le i\le s\}$. Then we have
$A_i\in\mL(aQ)\subset K$ for all $1\le i\le s$. Since $Q$ is a
rational place, one can find an $\F_q$-basis $y_1,\ldots,y_r$ of
$\mL(aQ)$ such that the pole orders $-\nu_Q(y_j)$ are strictly
increasing.

Thus, $A_i$ can be expressed as $\sum_{i=1}^ra_{ij}y_j$ for some $a_{ij}\in\F_q$. We rewrite ~\eqref{eq:3.1} into the following identity
\begin{equation}\label{eq:3.2}
\sum_{j=1}^{r}\left(\sum_{i=1}^{s}a_{ij}f_i\right)y_j=0.
\end{equation}
As $\sum_{i=1}^{s}a_{ij}f_i\in \mV\subseteq \mL(D)$,  and $\mV \cap \mL(D-Q')=\{0\}$, we know that either $\sum_{i=1}^{s}a_{ij}f_i=0$  or
\[\nu_{Q'}\left(\sum_{i=1}^{s}a_{ij}f_i\right)=0, \quad \mbox{and hence}\quad \nu_{Q'}\left(\left(\sum_{i=1}^{s}a_{ij}f_i\right)y_j\right)=\nu_Q(y_{j})e(Q'|Q), \] where $e(Q'|Q)$ denotes the ramification index of $Q'$ over $Q$. As the $\nu_{Q'}(y_j)$ for $j=1,2,\dots,s$ are distinct, this implies that $\sum_{i=1}^{s}a_{ij}f_i=0$ for all $1\le j\le r$. Therefore, $a_{ij}=0$ for all $1\le i\le s$ and $1\le j\le r$ since $f_1,f_2,\dots,f_s$ are $\F_q$-linearly independent. So $A_1=\cdots=A_s=0$. This is a contradiction and the proof is completed.
\end{proof}

The above lemma provides a sufficient condition  under which $\F_q$-linear independence of a set $f_1,\ldots,f_s\in \mL(D)$ of functions is equivalent to $K$-linear independence. Now we give an alternative condition although we will mainly use Lemma \ref{lem:3.1} in this paper.

\begin{lemma}\label{lem:3.2}
Let $F/K$ be a finite field extension of degree $n<+\infty$.   Suppose that there exists a rational place $Q$ in $K$ such that
there is only one place $Q'$ of $F$ lying above $Q$. Let $D$ be a positive divisor of $F$ with $Q'\not\in\supp(D)$ and $\deg(D)< n$.
Then $f_1,\ldots,f_s\in \mL(D)$ are $\F_q$-linearly independent if and only if they are  linearly independent over $K$.
\end{lemma}
\begin{proof} The ``if" part is clear. Now assume that $f_1,\ldots,f_s\in \mL(D)$ are $\F_q$-linearly independent.
Suppose that there would exist functions $A_1,\ldots,A_s\in K$ such that not all $A_i$ were not zero and
\begin{equation}\label{eq:3.3}
\sum_{i=1}^{s}A_if_i=0.
\end{equation}
 We are going to derive a contradiction.

As in the proof of Lemma \ref{lem:3.1}, we may assume that $\nu_P(A_i)\ge 0$ for all places $P\neq Q$ of $K$.  Let $a=\max\{-\nu_Q(A_i):\; A_i\neq0,\; 1\le i\le s\}$. Then we have $A_i\in\mL(aQ)\subset K$ for all $1\le i\le s$. Since $Q$ is a rational place, one can  find an $\F_q$-basis  $y_1,\ldots,y_r$ of $\mL(aQ)$  such that the pole orders $-\nu_Q(y_j)$ are strictly  increasing as $j$ increases from $1$ to $r$.

Thus, $A_i$ can be expressed as $\sum_{i=1}^ra_{ij}y_j$ for some $a_{ij}\in\F_q$. We rewrite ~\eqref{eq:3.3} into the following identity
\begin{equation}\label{eq:3.4}
\sum_{j=0}^{r}\left(\sum_{i=1}^{s}a_{ij}f_i\right)y_j=0.
\end{equation}
Assume that $b$ is the largest index such that $\sum_{i=1}^{s}a_{ib}f_i\neq0$. Such an index must exist as not all $a_{ij}$'s are $0$, and $f_1,f_2,\dots,f_s$ are linearly independent over $\F_q$.
Then the above identity becomes
\begin{equation}\label{eq:3.5}
-\sum_{j=0}^{b-1}\left(\sum_{i=1}^{s}a_{ij}f_i\right)y_j=\left(\sum_{i=1}^{s}a_{ib}f_i\right)y_b.
\end{equation}
Since $Q'$ is the unique place lying above $Q$, we have $e(Q'|Q)\deg(Q')=n$.   Then, the fact that $\sum_{i=1}^{s}a_{ij}f_i\in \mL( D)$ implies that either $\sum_{i=1}^{s}a_{ij}f_i=0$ or
 $\nu_{Q'}\left(\sum_{i=1}^{s}a_{ij}f_i\right)\leq \frac{\deg(D)}{\deg(Q')}<e(Q'|Q)$.
Therefore, the right hand side of \eqref{eq:3.5} gives
\begin{eqnarray*}\nu_{Q'}\left(\left(\sum_{i=1}^{s}a_{ib}f_i\right)y_b\right)&\le& \frac{\deg(D)}{\deg(Q')}+\nu_{Q}(y_b)e(Q'|Q)\\
&<& e(Q'|Q)+(\nu_{Q}(y_{b-1})-1)e(Q'|Q)=\nu_{Q}(y_{b-1})e(Q'|Q),\end{eqnarray*}
while the left hand side of \eqref{eq:3.5} gives
\[\nu_{Q'}\left(-\sum_{j=0}^{b-1}\left(\sum_{i=1}^{s}a_{ij}f_i\right)y_j\right)\ge \min_{1\le j\le b-1}\nu_{Q'}(y_j)=\nu_{Q}(y_{b-1})e(Q'|Q).\]
This is a contradiction and the proof is completed.
\end{proof}

\begin{rmk}\label{rmk:2}{\rm The requirement of $\deg(D)<[F:K]=n$ in Lemma \ref{lem:3.2} makes it difficult to compute the dimension of $\mL(D)$ as the genus $\g$ of $F$ is usually larger than $n$. While in Lemma \ref{lem:3.1}, there is no such a requirement. When $\deg(D-Q')\ge 2\g-1$ and $\mV\oplus\mL(D-Q')=\mL(D)$, then by the Riemann-Roch theorem we have $\dim_{\F_q}(\mV)=\dim_{\F_q}(D)-\dim_{\F_q}(D-Q')=\deg(Q')$.
}\end{rmk}

For each element $\sigma\in\Aut(F/\F_q)$, denote by $F^{\Gs}$ the fixed field by $\langle \sigma \rangle$, i.e., $F^{\Gs}=\{x\in F:\; x^{\Gs}=x\}$.  By the Galois theory, if $\Gs$ has a finite order, then $F/F^{\Gs}$ is a Galois extension and  $\Gal(F/F^{\Gs})=\langle \sigma \rangle$.
\begin{defn}
(Moore Matrix) Let $F/\F_q$ be a field extension. Let $f_1,\ldots,f_s$ be elements of $F$, the Moore Matrix is defined by
$M(f_1,\ldots,f_n)=\begin{pmatrix}
    f_1 & \cdots & f_s \\
    f_1^q & \cdots & f_n^q \\
    \vdots & \ddots & \vdots \\
    f_1^{q^{s-1}} & \cdots & f_s^{q^{s-1}} \\
\end{pmatrix}$,

%
It is a well-known fact that $f_1,\ldots,f_s$ are linearly independent over $\F_q$ if and only if the Moore Determinant $\det(M(f_1,\ldots,f_s))$ is nonzero.
\end{defn}

Now we generalize the above Moore matrix as follows.
\begin{defn}
\label{def:auto-moore}
{\em(Automorphism Moore Matrix)} Let $F/\F_q$ be a field extension. Let $\Gs$ be an automorphism in $\Aut(F/\F_q)$.  Let $f_1,\ldots,f_s$ be elements of $F$.  The $\sigma$-Moore matrix $M_{\sigma}(f_1,\ldots,f_s)$ is defined by
$M_\sigma(f_1,f_2,\dots,f_s) = \begin{pmatrix}
  f_1 & \cdots & f_s \\
                             f_1^\sigma & \cdots & f_s^\sigma \\
                             \vdots & \ddots & \vdots \\
                             f_1^{\sigma^{s-1}} & \cdots & f_s^{\sigma^{s-1}} \\
\end{pmatrix}$.
\end{defn}
\begin{rmk}\label{rmk:1}{\rm
If $\Gs$ is the usual Frobenius, i.e., $f^{\Gs}=f^q$ for all $f\in F$. Then  we have that $\det(M_\sigma(f_1,\ldots,f_s))\neq 0$ if and only if $f_1,\ldots,f_n$ are linearly independent over $F^{\Gs}=\F_q$.}
\end{rmk}
Our next theorem can be seen as a generalization of the result given in Remark \ref{rmk:1}. 
\begin{lemma}\label{lem:3.3} Let $\Gs\in \Aut(F/\F_q)$.
Let $f_1,\ldots,f_s\in F$.
Then the $\sigma$-Moore determinant $\det(M_{\sigma}(f_1,f_2,\ldots,f_s))$ equals $0$ if and only if $f_1,\ldots,f_n$ are linearly dependent over $F^{\Gs}$.
\end{lemma}
\begin{proof}
Let us prove the ``if" part first.  Assume that $f_1,\ldots,f_s$ are linearly dependent over $F^{\Gs}$, then there exist functions
$A_1,\ldots,A_s\in F^{\Gs}$ such that $A_1,\ldots,A_s\in F^{\Gs}$ are not all zero and
\begin{equation}\label{eq:x2}
A_1f_1+\ldots+A_sf_s=0.
\end{equation}
For each $0\le i\le s-1$, let automorphism $\sigma^i$ act on both the sides of \eqref{eq:x2}, then we have
\begin{equation}\label{eq:x0}
A_1f_1^{\sigma^i}+\ldots+A_sf_s^{\sigma^i}=0.
\end{equation}
Note that in the above equation, we use the fact that $A_j^{\Gs^i}=A_j$.
The equation \eqref{eq:x0} implies that $(A_1,\ldots,A_s)$ is a nonzero solution of  $M_{\sigma}(f_1,\ldots,f_s)\bz={\bf 0}$. Hence, we conclude that $\det(M_{\sigma}(f_1,f_2,\ldots,f_s))=0$.

Next we prove the ``only if" part by induction. It is clearly true for the case where $s=1$. Now assume that it holds for $s-1$. Suppose that
$\det(M_{\sigma}(f_1,f_2,\ldots,f_s))= 0$ and $f_1,f_2,\ldots,f_s$ are linearly independent over $F^{\Gs}$. We will derive a contradiction.

As $\det(M_{\sigma}(f_1,f_2,\ldots,f_s))= 0$,
 there exist  $A_1,\ldots,A_s\in F$ such that not all $A_1,\ldots,A_s$ are equal to $0$ and
 \begin{equation*}
 A_1f_1^{\sigma^{i}}+\ldots+A_sf_s^{\sigma^i}=0, \text{  for all}\; i\in \{0,\ldots,s-1\}.
 \end{equation*}
 Without loss of generality, we may assume that $A_1\neq 0$. Let $B_i=\frac{A_i}{A_1}\in F$ and we have
 \begin{equation}\label{eq:x3}
 f_1^{\sigma^{i}}+B_2f_2^{\sigma^i}\ldots+B_sf_s^{\sigma^i}=0, \text{  for }i\in \{0,\ldots,s-1\}.
 \end{equation}
Let $\sigma$ acts on both the sides of \eqref{eq:x3}, then
 \begin{equation}\label{eq:x4}
 f_1^{\sigma^{i+1}}+B_2^\sigma f_2^{\sigma^{i+1}}\ldots+B_s^\sigma f_s^{\sigma^{i+1}}=0, \text{  for }i\in \{0,\ldots,s-2\}.
 \end{equation}
By subtracting the $i$-th equation in \eqref{eq:x4} from the $(i+1)$-th equation in \eqref{eq:x3}, we obtain
\begin{equation}\label{eq:3.8}
(B_2-B_2^{\Gs}) f_2^{\sigma^{i+1}}+\ldots+(B_s-B_s^{\Gs})f_s^{\sigma^{i+1}}=0, \text{  for }i\in \{0,\ldots,s-2\}.
 \end{equation}
As $f_2,\dots,f_s$ are linearly independent over $F^\sigma$, by induction hypothesis, we have
\[ \det(M_\sigma(f_2,\dots,f_s)) \neq 0 , \ ~~ \text{which implies} ~~ \det(M_\sigma(f_2^\sigma,\dots,f_s^\sigma)) = \bigl( \det(M_\sigma(f_2,\dots,f_s))\bigr)^\sigma \neq 0 \ . \]
But then the linear dependence \eqref{eq:3.8} implies that $B_i-B_i^\sigma=0$ for all $2\le i\le s$. Thus, $B_i\in F^{\Gs}$ and \eqref{eq:x3} gives a non-trivial linear dependence of $f_1,f_2,\ldots,f_s$ over $F^{\Gs}$, a contradiction.
\end{proof}

%
Combining Lemmas \ref{lem:3.1} with \ref{lem:3.3} gives the following.


\begin{cor}\label{cor:3.4} Assume that the conditions in Lemma {\rm \ref{lem:3.1}} are satisfied with $K=F^{\Gs}$. Then for $f_1,f_2,\dots,f_s\in\mV\subseteq \mL(D)$, the  $\sigma$-Moore determinant $\det(M_{\sigma}(f_1,f_2,
\dots,f_s))= 0$ if and only if $f_1,f_2,\dots,f_s$ are linearly dependent over $\F_q$.
\end{cor}

Combining Lemmas \ref{lem:3.2} with \ref{lem:3.3} gives the following.
\begin{cor}\label{cor:3.5} Assume that the conditions in Lemma {\rm \ref{lem:3.2}} are satisfied with $K=F^{\Gs}$. Then for $f_1,f_2,\dots,f_s\in \mL(D)$, the  $\sigma$-Moore determinant $\det(M_{\sigma}(f_1,f_2,\ldots,f_s))= 0$ if and only if $f_1,f_2,\dots,f_s$ are linearly dependent over $\F_q$.
\end{cor}

\begin{rmk}\label{rmk:3} In \cite{GK-combinatorica},
the function field $F$ is the rational function $\F_q(x)$. The
automorphism $\Gs\in\Aut(F/\F_q)$ is given by $x\mapsto \Gg x$, where
$\Gg$ is a primitive element of $\F_q^*$. It is clear that the order
of $\Gs$ is $q-1$. The fixed field $F^{\Gs}$ is $\F_q(x^{q-1})$. Thus,
the degree $[F:F^{\Gs}]$ of extension $F/F^{\Gs}$ is $q-1$. Now for
$m<q-1$, we consider the Riemman-Roch space $\mL((m-1)\Pin)$, where
$\Pin$ is the unique pole of $x$. Then $\mL((m-1)\Pin)$ in fact
consists of all polynomials in $\F_q[x]$ of degree at most $m-1$. It
is clear that $((m-1)\Pin)^{\Gs}=(m-1)\Pin$. Furthermore, the rational
place $y-\Gg$ of $F^{\Gs}$ is fully inert in $F$, where
$y=x^{q-1}$. This is because $x^{q-1}-\Ga$ lies over $y-\Ga$ and
$x^{q-1}-\Ga$ has degree $q-1$. Thus, all conditions in Lemma
\ref{lem:3.2} are satisfied.  Therefore, by Corollary \ref{cor:3.5}
for a set of polynomials $f_1,f_2,\dots,f_s$ in $\F_q[x]$ of degree at
most $m-1$, the $\sigma$-Moore determinant
$\det(M_{\sigma}(f_1,f_2,\ldots,f_s))= 0$ if and only if
$f_1,f_2,\dots,f_s$ are linearly dependent over $\F_q$. This is exact
the result of Lemma 12 of \cite{GK-combinatorica}. Note that the Moore
determinant is called a folded Wronskian determinant in
\cite{GK-combinatorica}.
\end{rmk}

\subsection{Construction}
 Let $\Gs\in\Aut(F/\F_q)$ be an automorphism of a finite order. Let
 $D$ be a divisor of $F$ such that $D^{\Gs}=D$. Assume that all the
 conditions in Lemma \ref{lem:3.1} are satisfied. Recall
 $\mV\subseteq\mL(D)$ such that $\mV\cap\mL(D-Q)=\{0\}$.

For each place $P\in \PP_F$  such that $P\not\in\supp(D)$ and $P,P^{\Gs^{-1}},\dots, P^{\Gs^{-(t-1)}}$ are distinct, we define the subspace $\mH_P$:
\begin{equation}
\label{eq:ag-subspace}
\mH_P=\{f\in \mV: f(P^{\sigma^{-i}})=0 \text{  for each }i\in \{0,\ldots,t-1\}\}=\mV\cap \mL\left(D-\sum_{i=0}^{t-1}P^{\sigma^{-i}}\right).
\end{equation}
Recall that $f(P)$ is defined to be the residue class of $f$ in the residue field $O_P/P$. Hence, it is clear that
\[\dim_{\F_q}(\mH_P)\ge \dim_{\F_q}(\mV)+\dim_{\F_q}\left(D-\sum_{i=0}^{t-1}P^{\sigma^{-i}}\right)  -\dim_{\F_q}(D)\ge \dim_{\F_q}(\mV)-t\deg(P).\]
 Let $f(P)^{\Gs}=f(P)^{q^e}$ for some integer $e\ge 0$. Thus, we have
$f^{\Gs^i}(P^{\Gs^i})=f(P)^{\Gs^i}=f(P)^{q^{ei}}$ for all integers $i\ge 0$.

Define $S_P=\{P^{\sigma^{-i}}:\; i\in \{0,\ldots,t-1\}\}$, and denote by $\mF_r$ a set of places $P$ with degree $r$ such that $S_P$ are disjoint and $|S_P|=t$.

\begin{theorem}\label{thm:3.6} For any integers $s, t$ with $1\le s\le t$,
the collection $(\mH_P)_{P\in \mF_r}$ of subspaces of $\mV$, each of codimension at most $rt$, is an $\left(s,\frac{\ell s}{r(t-s+1)}\right)$ strong subspace design, where $\ell=\deg(D)$.
\end{theorem}
\begin{proof}
Let $\mW\subseteq \mV$ be an $\F_q$-subspace of dimension $s$. Let $f_1,\ldots,f_s$ be a basis for $\mW$.  Denote the dimension $\dim_{\F_q}(\mW\cap \mH_P)$ by $d_P$. Let $\{g_1,\dots,g_{d_P}\}$ be a basis of  $\mW\cap \mH_P$. Extend this basis to a basis $\{g_1,\dots,g_{d_P}, g_{d_P+1},\dots,g_s\}$ of $\mW$. Then it is clear that
\begin{equation}\label{eq:3.10}
\det(M_\sigma(f_1,\ldots,f_s))=b\det(M_\sigma(g_1,\ldots,g_s))
\end{equation}
for some $b\in \F_q^*$.

For any $g\in \mW\cap \mH_P$ and any $i, j$ with $0\le i\le s-1$ and $0\le j\le t-s$, we have $g(P^{\Gs^{-(i+j)}})=0$, i.e.,
\begin{equation}\label{eq:3.11}0=(g(P^{\Gs^{-(i+j)}}))^{q^{ei}}=(g(P^{\Gs^{-(i+j)}}))^{\Gs^{i}}=g^{\Gs^i}(P^{\Gs^{-j}}).\end{equation}
By definition of  determinants, we have $\det(M_\sigma(g_1,\ldots,g_s)(P^{-j}))=\sum_{\pi\in S_s}{\rm sgn}(\pi)\prod_{i=0}^{s-1}g_{\pi(i)}^{\Gs^{i}}(P^{-j})$,
where $S_s$ is the symmetric group. By \eqref{eq:3.11}, $\nu_{P^{-j}}(g_{\pi(i)})\ge 1$ whenever $\pi(i)\in\{1,\dots,d_P\}$. This implies that $\nu_{P^{-j}}\left({\rm sgn}(\pi)\prod_{i=0}^{s-1}g_{\pi(i)}^{\Gs^{u}}\right)\ge d_P$ for all $\pi$. Hence, $\nu_{P^{-j}}(M_\sigma(g_1,\ldots,g_s))\ge d_P$ for all $j\in\{0,1,\dots,t-s\}$. In conclusion, we have $M_\sigma(f_1,\ldots,f_s)\in\mL\left(sD-\sum_{P\in\mF_r}\sum_{j=0}^{t-s}d_PP^{-j}\right)$.

As $M_\sigma(f_1,\ldots,f_s)$ is a nonzero function, we must have
\begin{equation*}
\ell s=\deg(sD)\geq \sum_{P\in \mF_r}d_Pr(t-s+1)\geq \sum_{P\in \mF_r}r(t-s+1)\dim(\mW\cap \mH_P).
\end{equation*}
The desired result follows.
\end{proof}
So far in this subsection, we made use of Lemma \ref{lem:3.1} and Corollary \ref{cor:3.4} for construction of subspace designs. We can also make use of Lemma \ref{lem:3.2} and Corollary \ref{cor:3.5} to construct subspace designs.
Let $D$ be a positive divisor  of $F$ such that $D^{\Gs}=D$ and $\deg(D)<[F:F^{\Gs}]$.
For each place $P\in \PP_F$  such that $P\not\in\supp(D)$ and $P,P^{\Gs^{-1}},\dots, P^{\Gs^{-(t-1)}}$ are distinct, we define the subspace $\mI_P$:
\begin{equation}
\mI_P=\{f\in \mL(D): f(P^{\sigma^{-i}})=0 \text{  for each }i\in \{0,\ldots,t-1\}\}.
\end{equation}
We present the following result without proof as it is very similar to the one of Theorem~\ref{thm:3.6}.
\begin{theorem}\label{thm:3.7} For any integers $s, t$ with $1\le s\le t$,
the collection $(\mI_P)_{P\in \mF_r}$ of subspaces of $\mL(D)$, each of codimension at most $rt$, is an $\left(s,\frac{\ell s}{r(t-s+1)}\right)$ strong subspace design, where $\ell=\deg(D)$.
\end{theorem}

\subsection{Picking the places indexing the subspaces}
\label{subsec:picking-places}
To obtain a large set $\mF_r$ of places which define the subspaces in Theorems \ref{thm:3.6} and \ref{thm:3.7}, we consider those places $P$ that split completely in $F/F^{\Gs}$. Thus,  $P,P^{\Gs^{-1}},\dots,$ $ P^{\Gs^{-(t-1)}}$ are distinct as long as $t\le [F:F^{\Gs}]={\rm ord}(\Gs)$.
\begin{lemma}
\label{lem:splitting}
Let $P$ be a place of degree $r$ in $F$ with $\gcd(r,[F:F^{\Gs}])=1$. If $P$ is unramified in  $F/F^{\Gs}$, then $P$  splits completely in $F/F^{\Gs}$.
\end{lemma}
\begin{proof} Let $R$ be the place of $F^{\Gs}$ that lies under $P$, which has inertia degree $f(P|R)$. As $r=\deg(P)=f(P|R)\deg(R)$ and $f(P|R)|[F:F^{\Gs}]$, we must have $f(P|R)=1$ and $\deg(R)=r$. Since $P$ is unramified, the desired result follows.
\end{proof}
In view of the above result, we can choose $\mF_r$ as follows. Let $r$
be co-prime to $n:={\rm ord}(\Gs)$. Let $P_1,\dots,P_N$ be all
non-conjugate places of degree $r$ that are not ramified. Then for
each $i\in\{1,2,\dots,N\}$, $P_i,P_i^{\Gs},\dots,P_i^{\Gs^{n-1}}$ are
all distinct. Thus, we can form $\lfloor n/t\rfloor$ sets $S_{P_i},
S_{P_i^{\Gs^{-t}}},\dots, S_{P_i^{\Gs^{-t(\lfloor n/t\rfloor-1)}}}$
that are pairwise disjoint.  On the other hand, by \cite[Corollary
  5.2.10(a)]{stich-book} there are at least
$\frac{q^r}r-(2+7\g)\frac{q^{r/2}}r$ places of degree $r$, where $\g$
is the genus $F$.  Hence, if $r\gg\log_q(2+7\g)$ and not many places
of degree $r$ are ramified, we have roughly $\frac1{rt}q^r$ such sets
$S_P$. In fact, for our examples based on cyclotomic function fields in the next section, there are no
places of degree $r$ that are ramified.

\section{Subspace design from cyclotomic function fields}\label{sec:4}
In this section, we will present subspace design from the construction given in Section \ref{sec:3} by applying cyclotomic function fields. We start with the subspace design in a ambient space of smaller dimension.

\medskip \noindent {\bf The small dimension case.}
If $\deg(D)$ is smaller than $n=[F:F^{\Gs}]$ and $n$ is smaller than the genus $\g(F)$ of $F$,  in general it is hard to compute dimension of the Riemann-Roch space $\mL(D)$. Therefore, we cannot use the construction given in Theorem \ref{thm:3.7}. In this subsection, we apply Theorem \ref{thm:3.7} to the case where we can estimate the dimension of $\mL(D)$.

Let $F$ be the rational function field $\F_q(x)$. Let
$\Gs\in\Aut(F/\F_q)$ be given by $x\mapsto \Gg x$, where $\Gg$ is a
primitive element of $\F_q^*$. By Remark \ref{rmk:3} and Theorem
\ref{thm:3.7}, one can obtain the subspace design given in
\cite{GK-combinatorica}. Below we show that the subspace design given
in \cite{GK-combinatorica} can be realized by using cyclotomic
function fields.

Put $K=\F_q(x)$. Let $p_1(x)$ be a monic linear polynomial. For
instance, we can simply take $p_1(x)=x$. Then the cyclotomic function
field $F_1:=K(\Lambda_{p_1})$ is a cyclic extension over $K$ with
$\Gal(F_1/K)\simeq\F_q^*$. In fact, $F_1=K(\lambda)=\F_q(\Gl)$ with
$\lambda$ satisfying $\lambda^{q-1}+x=0$. Thus,
$K=\F_q(\Gl^{q-1})$. Let $\Gg$ be a primitive root of $\F_q$ and let
$\Gs\in\Gal(F/K)$ be defined by $\Gl^{\Gs}=\Gl^{\Gg}=\Gg\Gl$. This
gives the exactly the same function fields and automorphism $\Gs$ as
in Remark \ref{rmk:2}. Therefore, we conclude that this cyclotomic
function field also realizes the subspace design given in
\cite{GK-combinatorica}.

Next we consider a monic \emph{primitive} quadratic polynomial
$p_2(x)=x^2+\Ga x+\Gb$ with $\Ga,\Gb\in\F_q$.  Then the cyclotomic
function field $F_2:=K(\Lambda_{p_2})$ is a cyclic extension over $K$
with $\Gal(F_2/K)\simeq(\F_q[x]/(p_2)^*$. In fact, $F_2=K(\lambda)$
with $\lambda$ satisfying $\Gl^{q^2-1}+\Gl^{q-1}(x^q+x+\Ga)+x^2+\Ga
x+\Gb=0.$ (see \cite{MXY16}).  Let $\Gs$ be a generator of
$\Gal(F_2/K)$. Then by the Galois theory, the fixed field $F_2^{\Gs}$
is the rational function field $K=\F_q(x)$. The genus of the function
field $F_2$ is $\g(F_2)=\frac{(q-2)(q+1)}{2}$ \cite{hayes,MXY16}.

The zero  of $p_2(x)$ is the unique  ramified place in $\F_q(x)$ and it is totally ramified. Let $P'$ be the unique place of $F_2$ that lies over the zero of $p_2(x)$. Let $\ell$ be an even positive integer with $\ell<q^2-1$ and let $D=(\ell/2)P'$. Then $\deg(D) =\ell$ and $D^{\Gs}=D$. Furthermore, we know that the  the zero of $(x-\Ga)$ is fully inert in $F_2/K$.  Thus, all the conditions in Lemma \ref{lem:3.2} are satisfied. By Theorem \ref{thm:3.7}, we have the following result.
\begin{theorem}\label{thm:4.1}
For all positive integers $s,r,t,m$ and prime powers $q$ satisfying $s\leq t\leq m=\zeta q^2$ for some $\zeta\in(0,1/2]$, the above construction yields a collection of
$M=\Omega(\frac{q^r}{rt})$ spaces $\mI_1,\ldots,\mI_M\subset \F_q^m$, each of codimension $rt$, which forms an $\left(s',\frac{(1+1/(2\zeta))ms'}{r(t-s'+1)}\right)$ strong subspace design for all $s' \le s$.
\end{theorem}
\begin{proof}
Choose $\ell$ such that the dimension of $\mL((\ell/2) P')$ is $m=\zeta q^2$. By the Riemman-Roch Theorem, we have $\zeta q^2\ge \deg( (\ell/2) P')-\g(F_2)+1$, i.e.,  $\ell\leq \zeta q^2+g-1\leq (1/2+\zeta)q^2$. The desired result follows from Theorem \ref{thm:3.7}.
\end{proof}

\smallskip \noindent {\bf The large dimension case.}
In this subsection, we will make use of Theorem \ref{thm:3.6} due to large genus. Let $p(x)\in\F_q[x]$ be a monic primitive polynomial of degree $d\ge 2$. Consider the cyclotomic function field $F:=K(\Lambda_{p(x)})$, where $K$ is the rational function field $\F_q(x)$. Then $F/K$ is a Galois extension with $\Gal(F/K)\simeq(\F_q[x]/(p(x)))^*$. Thus, $\Gal(F/K)$ is a cyclic group of order $q^d-1$. Let $\Gs$ be a generator of this group. Then by the Galois theory, the fixed field $F^{\Gs}$ is the rational function field $\F_q(x)$.

The zero  of $p(x)$ is the unique  ramified place in $\F_q(x)$ and it is totally ramified. Let $P'$ be the unique place of $F$ lying over the zero of $p(x)$. Let $Q's$ be the unique place of $F$ that lies over the zero of $x$. Since $Q'$ is totally inert, we have $\deg(Q')=[F:F^{\Gs}]=q^d-1:=m$.

The genus of the function field $F$ is $\g=\frac12\left(d-2+\frac{q-2}{q-1}\right)(q^d-1)+1$. Put
$D=\left\lceil\frac{2\g+m-1}d\right\rceil P'.$ Then $\ell=\deg(D)\ge 2\g+m$ and hence, $\dim_{\F_q}(D-Q')=\deg(D-Q')-\g+1$. Choose $\mV\subseteq\mL(D)$ such that $\mV$ and $\mL(D-Q')$ are a direct sum of $\mL(D)$. Thus, we have $\mV\cap \mL(D-Q')=\{0\}$ and $\dim_{\F_q}(\mV)= \dim_{\F_q}(D)- \dim_{\F_q}(D-Q')=q^d-1=m$.

Thus, all the conditions in Lemma \ref{lem:3.1} are satisfied. By Theorem \ref{thm:3.6}, we have the following.

%
\begin{theorem}\label{thm:4.2}
For all positive integers $s,r,t,d,m$ and prime powers $q$ satisfying $\gcd(r,m)=1$ and $s\leq t\leq m/r=(q^d-1)/r$, there is an explicit collection of
$M=\Omega(\frac{m \cdot q^r}{rt})$ spaces $\mH_1,\ldots,\mH_M\subset \F_q^m$, each of codimension at most $rt$, which forms an $(s',\frac{(d-1/(q-1))ms'}{r(t-s'+1)})$-strong subspace design for all $s' \le s$. Furthermore, the subspace design can be constructed in $\text{poly}(q,m,r)$ time.
\end{theorem}
\begin{proof}
  The subspace design property follows from Theorem \ref{thm:3.6}
  since $\ell=\deg(D)\le (d-1/(q-1))m$. The construction of the
  subspace design mainly involves finding a basis of $\mV$ and
  evaluations of functions at places of degree $r$. We have described
  how to compute a basis in Lemma \ref{lem:2.5} and how to evaluate a
  function a high degree place in Lemma \ref{lem:2.4}. The places of degree $r$ defining the subspaces in the subspace design can be computed as described in Section~\ref{subsec:picking-places}. We can enumerate over all degree $r$
  irreducible polynomials $R \in \F_q[x]$ by brute-force in $q^{O(r)}$
  time. None of these places are ramified, and by Lemma~\ref{lem:splitting} each of these places $R$ splits completely into $m$ places of degree $r$, say $\{P^{\Gs^{i-1}} \mid 1 \le i \le m\}$, in $F$.  So we can pick $b=\lfloor \frac{m}{t}\rfloor$ of these places $P,P^{\Gs^t},\dots,P^{\Gs^{(b-1)t}}$, and
 define a particular subspace of co-dimension $rt$ associated with each of them as in
  \eqref{eq:ag-subspace}.
\end{proof}

\smallskip
\noindent
By setting $t\approx 2s$ and $r\approx \lfloor \frac{\epsilon m}{2s}\rfloor$ in Theorem \ref{thm:4.2}, we obtain the Main Theorem \ref{thm:1.1}.


\bibliographystyle{alpha}
\bibliography{sd}


\end{document}